\documentclass[11pt]{article}

\usepackage[bibliosources=refs.bib]{pomegranate}

\DeclareOperator{\Cov}
\DeclareOperator{\Var}
\DeclareOperator{\diag}
\DeclareOperator{\Unif}
\newcommand{\cD}{\mathcal D}
\newcommand{\cG}{\mathcal G}
\newcommand{\relint}{\operatorname{relint}}

\title{On Rounding on the Hypersimplex}
\author[1]{Nima Anari}
\author[1]{Alireza Haqi}
\author[1]{Eric Ma}
\affil[1]{Stanford University\\
\texttt{\{anari,ahaqi,eryma\}@stanford.edu}}
\date{}

\begin{document}

\maketitle

\begin{abstract}
We study correlated rounding on the hypersimplex, the base polytope of the
uniform matroid.  For each point \(x\) in the hypersimplex, the goal is to
sample a \(k\)-subset \(A(x)\) with marginals \(x\), while coupling the samples
for all choices of \(x\) so that nearby inputs produce nearby sets.  We give
conditional constant-stretch results for the natural maximum-entropy
sequential scheme, based on a conjectured spectral property of the covariance
matrix of the maximum-entropy distribution over \(k\)-subsets; this conjecture
has been extensively tested numerically.  Under this property, the scheme samples the
maximum-entropy \(k\)-subset distribution with prescribed marginals using a
common random ordering and common uniform thresholds.  For every
\(x,y\in[0,1]^n\) with
\(\sum_i x_i=\sum_i y_i=k\), it satisfies
\[
  \E*{\abs{A(x)\triangle A(y)}} \le 6\norm{x-y}_1.
\]
Thus, conditional on the spectral hypothesis, the previous \(O(\log k)\) bound
for hypersimplex correlated rounding would improve to a constant and the open
question raised by Naor, Raju, Shetty, Srinivasan, Valieva, and Wajc would have
a positive answer.  By adding dummy coordinates, the same conditional result
gives stretch at most \(12\) for the at-most-\(k\) polytope.
\end{abstract}

\section{Introduction}

Randomized rounding is a central tool for converting fractional solutions into
discrete objects.  Given a point in a polytope, one samples an integral point
whose coordinates have the prescribed expectations.  This lets an algorithm
reason about a convex relaxation and then return a feasible discrete solution.
Many rounding methods also preserve hard constraints, such as a cardinality or
matroid-base constraint, and give negative-dependence or concentration
guarantees
\cite{AgeevSviridenko2004,GandhiKhullerParthasarathySrinivasan2006,
ChekuriVondrakZenklusen2010}.

Correlated rounding strengthens this pointwise view.  Instead of choosing a
separate rounding distribution for each fractional point, it asks for a single
random map that rounds all fractional points at once.  The map should preserve
marginals at every input, and it should be stable: nearby fractional points
should round to nearby discrete objects.  Such consistency is important in
online, dynamic, and multi-scenario settings, where an algorithm may encounter
many nearby fractional solutions and the rounded solution should not be
resampled almost independently after a small perturbation.

Let
\[
  \Delta_{n,k}^{=}
  =
  \set*{x\in[0,1]^n \given \sum_{i=1}^n x_i=k}
\]
be the \(k\)th hypersimplex, equivalently the base polytope of the rank-\(k\)
uniform matroid.  A correlated rounding scheme for \(\Delta_{n,k}^{=}\) is a
random map \(x\mapsto A(x)\), where \(A(x)\subseteq[n]\) has size \(k\) and
\(\P{i\in A(x)}=x_i\) for every \(i\).  Its stretch is the smallest
\(\alpha\) for which
\[
  \E*{\abs{A(x)\triangle A(y)}}\le \alpha\norm{x-y}_1
  \qquad
  \forall x,y\in\Delta_{n,k}^{=}.
\]
Thus stretch is the Lipschitz constant of the random rounding map, where
distance between rounded sets is measured by symmetric difference.

The case \(k=1\) is correlated sampling from the probability simplex.  In that
case the goal is to sample from nearby distributions with as much overlap as
possible, and the problem has been studied since the statistical sampling work
of Keyfitz \cite{Keyfitz1951}.  In theoretical computer science, simplex
correlated sampling is a useful primitive in approximation algorithms,
parallel repetition, and locality-sensitive hashing
\cite{Broder1997,KleinbergTardos2002,Holenstein2007}.

For \(k>1\), the hypersimplex captures the simplest cardinality constraint:
round to exactly \(k\) selected elements while preserving the individual
marginals.  This is the uniform-matroid version of dependent rounding.  The
output must always have exactly \(k\) elements, and the same randomness must
work coherently for every input.  Chen, Di Castro, Karnin, Lewin-Eytan, Naor,
and Schwartz
\cite{ChenDiCastroKarninLewinEytanNaorSchwartz2017} introduced correlated
rounding for uniform matroids and obtained \(O(\log n)\) stretch.  Naor, Raju,
Shetty, Srinivasan, Valieva, and Wajc
\cite{NaorRajuShettySrinivasanValievaWajc2026} recently improved this to
\(O(\log k)\), independent of the ambient dimension, with efficient
implementations and applications to paging, metric multi-labeling, and
reallocation.  They also identified the remaining barrier: after explaining
that a better dimension-dependent bound would feed through their recurrence,
they asked whether any improvement beyond logarithmic stretch is possible, and
whether \(\Theta(\log k)\) is optimal.  We prove a constant bound for the exact
hypersimplex conditional on the following covariance spectral hypothesis, stated
precisely in \cref{conj:covariance-spectral}: for every covariance matrix \(C\)
of an external-field \(r\)-subset measure arising in the sampler, if
\(D=\diag(C_{ii})\), then \(D^{-1/2}CD^{-1/2}\) has one zero eigenvalue and all
other eigenvalues in \([1,2]\).

The algorithm is simple.  Given \(x\in\Delta_{n,k}^{=}\), let \(\mu_x\) be the
maximum-entropy distribution on \(k\)-subsets with marginals \(x\).  For
relative-interior points, this is the external-field measure
\[
  \mu_x(S)
  =
  \frac{\prod_{i\in S}\lambda_i}{e_k(\lambda)}
  \qquad (\abs{S}=k),
\]
where \(\lambda_i>0\) and \(e_k\) is the elementary symmetric polynomial of
degree \(k\).  This is the same distribution obtained by taking independent
weighted Bernoulli variables and conditioning on selecting exactly \(k\)
coordinates.  The maximum-entropy distribution is the natural choice here: it
imposes no structure beyond the one-coordinate marginals and the
exact-cardinality constraint.

To sample \(\mu_x\), the algorithm chooses a uniformly random permutation
\(\pi\) of \([n]\) and independent uniforms \(u_1,\ldots,u_n\).  It scans
coordinates in the order \(\pi\).  At each step it computes the conditional
marginal that the current coordinate belongs to the remaining set, and compares
that marginal to the corresponding uniform.  This is sequential
inverse-transform sampling.  For correlated rounding, the same \(\pi\) and the
same uniforms are used for every input \(x\).

The analysis must control stability under conditioning.  After some
coordinates have been exposed, a small change in \(x\) changes the external
field and hence all later conditional marginals.  These changes could, in
principle, accumulate over the scan.  The conditional theorem says that, under
the covariance spectral hypothesis, the maximum-entropy law and a random order
make this accumulation uniformly bounded.

\paragraph{Status of the main claim.}
The original version of this manuscript claimed an unconditional constant
bound.  That claim relied on the covariance spectral statement now isolated as
\cref{conj:covariance-spectral}.  The proof previously given for that
statement was incorrect: it attempted to obtain the normalized spectral bound
from a Hodge--Riemann/inertia argument, but that argument does not establish
the normalized spectrum assertion needed in the variance-drop step.  The
present version withdraws the unconditional claim and records the remaining
argument as a conditional theorem.

\begin{theorem}[Conditional exact hypersimplex]\label{thm:main}
Assume \cref{conj:covariance-spectral} holds for every external-field
\(r\)-subset measure that can arise during the sequential sampler.  Then, for
every \(1\le k\le n\), the maximum-entropy sequential rounding scheme described
above is a correlated rounding scheme for \(\Delta_{n,k}^{=}\) with stretch at
most \(6\).  That is, for every \(x,y\in\Delta_{n,k}^{=}\),
\[
  \E*{\abs{A(x)\triangle A(y)}}\le 6\norm{x-y}_1.
\]
\end{theorem}

The conditional theorem extends to the more common at-most-\(k\) polytope
\[
  \Delta_{n,k}^{\le}
  =
  \set*{x\in[0,1]^n\given \sum_{i=1}^n x_i\le k}
\]
by adding dummy coordinates that absorb the slack and then deleting those
coordinates after rounding.

\begin{corollary}[Conditional at most \(k\) elements]\label{cor:at-most}
Under the hypothesis of \cref{thm:main}, there is a correlated rounding scheme
for \(\Delta_{n,k}^{\le}\) which outputs a subset of \([n]\) of size at most
\(k\), preserves all marginals, and has stretch at most \(12\):
\[
  \E*{\abs{A(x)\triangle A(y)}}\le 12\norm{x-y}_1
  \qquad
  \forall x,y\in\Delta_{n,k}^{\le}.
\]
\end{corollary}

\paragraph{Proof idea.}
The proof below is conditional only on \cref{conj:covariance-spectral}; all
other reductions are proved from that input.  It studies infinitesimal
perturbations of the input.  Write
\(\theta_i=\log \lambda_i\).  If the parameter \(\theta\) moves with velocity
\(h\), then the marginal vector moves with velocity \(Ch\), where \(C\) is the
covariance matrix of \(\mu_x\).  This is the standard derivative formula for
an exponential family.  The same formula holds after conditioning on any
partial reveal: if \(R\) is the set of unrevealed coordinates and \(r\) is the
residual quota, the velocity of the current conditional marginal vector is
\(C_{R,r}h_R\).  Because the next coordinate is uniform in \(R\), the
instantaneous rate at which the two coupled samplers first disagree is
\[
  \E*{
    \sum_\ell \frac{1}{\abs{R_\ell}}
    \norm{C_{R_\ell,r_\ell}h_{R_\ell}}_1
  }.
\]

There are two steps.  First, for a fixed external field, a first disagreement
costs exactly \(2\) in final symmetric difference.  Indeed, after one run takes
the current coordinate and the other does not, the remaining quotas differ by
one.  Conditional marginals for external-field \(k\)-subset measures are
monotone in the quota, so the same future uniforms couple the two continuations
monotonically; their future sets differ in exactly one coordinate.  In the
infinitesimal comparison between \(x\) and \(x+dx\), the change in the future
external field after the first crossing affects only lower-order terms.

Second, we bound the first-disagreement rate.  Covariance matrices of
external-field \(k\)-subset measures are graph Laplacians: negative
correlations give nonnegative edge weights \(c_{ij}=-C_{ij}\), and the identity
\(\sum_i\1_i=k\) gives \(C\1=0\).  A coarea formula turns
\(\norm{C_{R,r}h_R}_1\) into an integral of cut variances.  A variance-drop
lemma says that revealing a uniformly random coordinate decreases each such
variance by at least a \(1/\abs{R}\) fraction.  This makes the whole
random-order sum telescope down to a static quantity
\[
  2\sum_{i<j} c_{ij}\abs{h_i-h_j},
  \qquad
  c_{ij}=-C_{ij}.
\]
The final ingredient is a static electrical bound.  For distinct \(a,b\), any
solution \(v\) of \(Cv=\1_a-\1_b\) satisfies
\[
  \sum_{i<j}c_{ij}\abs{v_i-v_j}\le 3.
\]
By decomposing an arbitrary zero-sum vector \(Ch\) into elementary transfers,
the static energy is at most \((3/2)\norm{Ch}_1\).  Hence the infinitesimal
first-disagreement rate is at most \(3\norm{Ch}_1\), and the final
infinitesimal stretch is at most \(6\norm{Ch}_1\).  Since \(Ch=\dot x\),
integrating along the line segment from \(x\) to \(y\) proves
\cref{thm:main}.

\paragraph{Organization.}
\Cref{sec:prelim} fixes notation, recalls the external-field measures used by
the sampler, and states the covariance spectral hypothesis.
\Cref{sec:algorithm} states the sequential sampler and proves exactness.
\Cref{sec:infinitesimal} converts stretch to an infinitesimal covariance
bound.  \Cref{sec:variance-drop} proves the dynamic variance-drop estimate
under the spectral hypothesis.  \Cref{sec:static} proves the static electrical
estimate.  \Cref{sec:proof-main} combines the ingredients and proves the
conditional main theorem and the at-most-\(k\) corollary.

\section{Preliminaries}\label{sec:prelim}

\paragraph{Notation.}
For \(R\subseteq[n]\), write \(\lambda_R=(\lambda_i)_{i\in R}\), and let
\[
  e_r(\lambda_R)
  =
  \sum_{\substack{T\subseteq R\\ \abs{T}=r}}
  \prod_{i\in T}\lambda_i
\]
be the elementary symmetric polynomial of degree \(r\), with
\(e_0(\lambda_R)=1\).  We write \(\1_i(S)=\1\{i\in S\}\) for the
\(i\)th coordinate indicator, and also use \(\1_i\) for the corresponding
standard basis vector when it appears in a linear equation.  For
\(A\subseteq R\), \(\1_A\) denotes its indicator vector.  For a vector \(a\),
\(a_R\) denotes its restriction to \(R\).

For \(R\subseteq[n]\), \(0\le r\le\abs{R}\), and positive weights
\(\lambda_R\), define the external-field \(r\)-subset measure
\[
  \mu_{R,r}^{\lambda}(T)
  =
  \frac{\prod_{i\in T}\lambda_i}{e_r(\lambda_R)}
  \qquad
  (T\subseteq R,\ \abs{T}=r).
\]
When the weights are clear, we write \(\mu_{R,r}\).  Its marginal vector is
\[
  p_i(R,r)
  =
  \P_{T\sim\mu_{R,r}}{i\in T}
  =
  \frac{\lambda_i e_{r-1}(\lambda_{R\setminus i})}
       {e_r(\lambda_R)}
  \qquad (i\in R).
\]
Let \(C_{R,r}\) be the covariance matrix of the indicators under
\(\mu_{R,r}\):
\[
  (C_{R,r})_{ij}
  =
  \Cov_{\mu_{R,r}}(\1_i,\1_j).
\]
One can view \(\mu_{R,r}\) as sampling \(r\) elements without replacement with
biases \(\lambda_i\).  Equivalently, take independent Bernoulli variables with
odds proportional to \(\lambda_i\), and condition on exactly \(r\) successes.
This family is closed under conditioning on whether a coordinate is selected,
which is why it is suited to a sequential sampler.

\paragraph{Maximum-entropy measures.}
The maximum-entropy distribution on \(k\)-subsets with marginals \(x\) is the
solution of
\[
  \max\set*{
    -\sum_{\abs{S}=k}\nu(S)\log\nu(S)
    \given
    \nu\in\cD_k,\quad
    \P_{S\sim\nu}{i\in S}=x_i\ \forall i
  },
\]
where \(\cD_k\) is the simplex of probability distributions on \(k\)-subsets.
This is a finite-dimensional convex optimization problem: entropy is strictly
concave on the support, and the constraints are linear.  For
\(x\in\relint(\Delta_{n,k}^{=})\), the optimizer has full support and the
Karush--Kuhn--Tucker conditions give
\[
  \nu(S)\propto \exp\!\parens*{\sum_{i\in S}\theta_i}
  =
  \prod_{i\in S}\lambda_i,
  \qquad
  \lambda_i=e^{\theta_i}>0.
\]
The vector \(\theta\) is unique up to adding the same constant to every
coordinate, since all sets have size \(k\).  Boundary points are obtained by
deleting coordinates with \(x_i=0\), contracting coordinates with \(x_i=1\),
and applying the same construction on the remaining relative interior.

The conditional argument uses the following covariance spectral hypothesis.

\begin{conjecture}[Covariance spectral hypothesis]\label{conj:covariance-spectral}
Let \(C=C_{R,r}\) be the covariance matrix of an external-field \(r\)-subset
measure with \(0<r<\abs{R}\), and assume no coordinate is deterministic.  Then:
\begin{enumerate}
  \item \(C\1=0\);
  \item \(C_{ij}\le0\) for all \(i\ne j\);
  \item if \(D=\diag(C_{ii})\), then
  \[
    B=D^{-1/2}CD^{-1/2}
  \]
  has one zero eigenvalue and all other eigenvalues in \([1,2]\).
\end{enumerate}
\end{conjecture}

The first two items are standard: \(C\1=0\) follows from the deterministic
identity \(\sum_{i\in R}\1_i(T)=r\), and pairwise negative correlation follows
because external-field \(r\)-subset measures are strongly Rayleigh
\cite{BorceaBrandenLiggett2009}.  The unproved part is the normalized spectral
assertion in item 3.  An earlier draft attempted to derive it from a
Hodge--Riemann/inertia argument for the Rayleigh-difference matrix, but that
derivation does not prove the normalized eigenvalue bound stated here.

\(C\) is a graph Laplacian.  We write
\[
  c_{ij}:=-C_{ij}\ge0\quad (i\ne j),
  \qquad
  d_i:=C_{ii}=\sum_{j\ne i}c_{ij}.
\]
The normalized eigenvalue bound in \cref{conj:covariance-spectral} says that
this covariance Laplacian has condition number at most \(2\) on the
degree-normalized subspace orthogonal to the zero eigenvector.  This is the
form used in the variance-drop argument below.

\section{The Sequential Sampler}\label{sec:algorithm}

Fix \(x\in\relint(\Delta_{n,k}^{=})\), and let
\(\mu_x=\mu_{[n],k}^{\lambda}\) be its maximum-entropy external-field measure.
The sampler uses a random permutation \(\pi\) and independent
\(u_i\sim\Unif[0,1]\).  Initially \(R=[n]\), \(r=k\), and \(A=\emptyset\).
When the next coordinate is \(i\in R\), set
\[
  q_i=p_i(R,r)
  =
  \frac{\lambda_i e_{r-1}(\lambda_{R\setminus i})}
       {e_r(\lambda_R)}.
\]
If \(u_i\le q_i\), add \(i\) to \(A\) and replace \(r\) by \(r-1\); otherwise
leave \(A\) unchanged.  In either case remove \(i\) from \(R\), and continue.
Once \(r=0\), all later coordinates are rejected; once \(r=\abs{R}\), all
later coordinates are accepted.
The random order is not needed for exact sampling, but it is important for the
stretch analysis: conditioned on the current state, the next coordinate is
uniform among the unrevealed coordinates.

\begin{lemma}[Exactness]\label{lem:exactness}
For every \(x\in\Delta_{n,k}^{=}\), the sequential sampler outputs a
\(k\)-subset distributed as \(\mu_x\).  In particular,
\(\P{i\in A(x)}=x_i\) for every coordinate \(i\).
\end{lemma}

\begin{proof}
It is enough to consider relative-interior points; boundary points follow by
the forced-coordinate reduction described above.  At a state \((R,r)\), the
conditional distribution of the unrevealed part under \(\mu_{R,r}\) is exactly
\[
  \mu_{R,r}(T)
  =
  \frac{\prod_{i\in T}\lambda_i}{e_r(\lambda_R)}
  \qquad (T\subseteq R,\ \abs{T}=r).
\]
The threshold \(p_i(R,r)\) is the conditional probability that \(i\) belongs to
this remaining set.  Conditioning on the outcome at \(i\), the remaining law is
\(\mu_{R\setminus i,r-1}\) if \(i\) is selected and
\(\mu_{R\setminus i,r}\) otherwise.  Thus the recursive procedure is exactly
the chain rule for the distribution \(\mu_{[n],k}\).
\end{proof}

\begin{lemma}[Quota monotonicity]\label{lem:quota-monotone}
For fixed \(R\), fixed positive \(\lambda_R\), and \(i\in R\), the conditional
marginal \(p_i(R,r)\) is nondecreasing in \(r\).
\end{lemma}

\begin{proof}
For \(0<r<\abs{R}\), the inequality \(p_i(R,r+1)\ge p_i(R,r)\) is equivalent,
after canceling the positive factor \(\lambda_i\), to
\[
  e_r(\lambda_{R\setminus i})^2
  \ge
  e_{r-1}(\lambda_{R\setminus i})e_{r+1}(\lambda_{R\setminus i}).
\]
This is Newton log-concavity for elementary symmetric polynomials with positive
variables, which also follows from the Lorentzian property of elementary
symmetric polynomials \cite{BrandenHuh2020}.
\end{proof}

\begin{lemma}[Cost of the first disagreement]\label{lem:first-disagreement}
Consider two sequential samplers with the same external field and the same
future order and uniforms.  Suppose that, at some coordinate \(i\), one run
selects \(i\) and the other does not, and before that step the two runs have
the same state.  If the two runs are coupled by the common future uniforms
thereafter, then the final sets differ in exactly two coordinates.
\end{lemma}

\begin{proof}
After the disagreement at \(i\), the remaining coordinate set is the same in
both runs, while the residual quotas differ by one.  By
\cref{lem:quota-monotone}, at every future state with the same remaining
coordinates, every conditional marginal for the larger quota is at least the
corresponding conditional marginal for the smaller quota.  Therefore the common
uniforms couple the larger-quota continuation so that it contains the
smaller-quota continuation.  The two future continuations have sizes that
differ by exactly one, so their future symmetric difference is exactly one.
Together with the coordinate \(i\), the total symmetric difference is exactly
two.
\end{proof}

\section{Infinitesimal Disagreement}\label{sec:infinitesimal}

Let \(x(t)\in\relint(\Delta_{n,k}^{=})\) be a smooth path, and choose a smooth
external-field parametrization
\[
  \theta_i(t)=\log\lambda_i(t).
\]
Since adding a common constant to all \(\theta_i(t)\) does not change the
measure, \(h=\dot\theta(0)\) is defined modulo multiples of \(\1\); all
expressions below are invariant under this ambiguity.  Let \(C=C_{[n],k}\) be
the covariance matrix under \(\mu_{x(0)}\).  If
\[
  H(S)=\sum_{i\in S}h_i,
\]
then
\begin{equation}\label{eq:marginal-velocity}
  \dot x_i(0)
  =
  \Cov_{\mu_{x(0)}}(\1_i,H)
  =
  (Ch)_i.
\end{equation}
Indeed, differentiating an expectation in the exponential family gives
\(\frac{d}{dt}\E_{\mu_t}{f}=\Cov_{\mu_t}(f,H)\).  Applying this with
\(f=\1_i\) gives \cref{eq:marginal-velocity}.  The same calculation applies
after any conditioning because conditional external-field measures have the
same form.  Thus, at state \((R,r)\), the derivative of the conditional
marginal vector is \(C_{R,r}h_R\).

For a fixed run of the base sampler at time \(0\), let \((R_\ell,r_\ell)\) be
the unrevealed set and residual quota just before the \(\ell\)th reveal.
Since the remaining order is random, the next revealed coordinate is uniform in
\(R_\ell\).
Define the infinitesimal first-disagreement functional
\begin{equation}\label{eq:V-def}
  \mathsf V(h)
  =
  \E*{
    \sum_\ell
    \frac{1}{\abs{R_\ell}}
    \norm{C_{R_\ell,r_\ell}h_{R_\ell}}_1
  },
\end{equation}
where the expectation is over the random order and the set sampled from the
base measure.  The factor \(1/\abs{R_\ell}\) appears because the next
coordinate is uniform in \(R_\ell\), and the \(\ell_1\) norm sums the threshold
changes over the possible next coordinates.

\begin{lemma}[Infinitesimal reduction]\label{lem:infinitesimal-reduction}
If, for every relative-interior point and every tangent direction,
\[
  \mathsf V(h)\le 3\norm{Ch}_1,
\]
then the sequential sampler has stretch at most \(6\) on
\(\Delta_{n,k}^{=}\).
\end{lemma}

\begin{proof}
Couple \(x(0)\) and \(x(\delta)\) using the same order and uniforms.  Until the
first disagreement, the two runs have the same revealed decisions, and the
threshold at state \((R,r)\) changes by
\[
  \delta (C_{R,r}h_R)_i+o(\delta)
\]
at coordinate \(i\).  Conditional on the common history, the probability that
the uniform \(u_i\) falls between the two thresholds is therefore
\[
  \delta\abs{(C_{R,r}h_R)_i}+o(\delta).
\]
Summing over the random next coordinate and then over reveal times gives
\[
  \P{\text{the two runs ever disagree}}
  \le
  \delta \mathsf V(h)+o(\delta).
\]
It remains to translate a first disagreement into final symmetric difference.
If, after the first threshold crossing, we freeze the external field at its
time-\(0\) value, then \cref{lem:first-disagreement} says that the final damage
is exactly \(2\).  In the actual \(x(\delta)\)-run, all later thresholds differ
from their time-\(0\) values by \(O(\delta)\), uniformly over the finitely many
states reached by a fixed relative-interior path.  Hence the conditional
expected extra future damage after the first crossing is \(O(\delta)\).  Since
the first crossing itself has probability \(O(\delta)\), this extra
contribution is \(O(\delta^2)=o(\delta)\).  Hence
\[
  \frac{d}{dt}
  \E*{\abs{A(x(t))\triangle A(x(0))}}\bigg|_{t=0}
  \le
  2\mathsf V(h)
  \le
  6\norm{\dot x(0)}_1,
\]
where \(\dot x(0)=Ch\) by \cref{eq:marginal-velocity}.  Integrating this
bound along the line segment from \(x\) to \(y\) gives the claimed stretch for
relative-interior points.  Boundary points follow by approximation and the
forced-coordinate reduction.
\end{proof}

\section{The Variance-Drop Bound}\label{sec:variance-drop}

This section proves the dynamic part of the argument:
\[
  \mathsf V(h)
  \le
  2\sum_{i<j}c_{ij}\abs{h_i-h_j},
\]
where \(c_{ij}=-C_{ij}\) for the initial covariance matrix \(C=C_{[n],k}\).
This reduces an adaptive quantity, accumulated over the whole sequential
exposure, to one static energy at the initial measure.

\begin{lemma}[Coarea for covariance Laplacians]\label{lem:coarea}
Let \(C=C_{R,r}\), let \(c_{ij}=-C_{ij}\), and let \(h\in\R^R\).  For
\(s\in\R\), put \(A_s=\set{i\in R\given h_i\ge s}\).  If
\(S\sim\mu_{R,r}\), then
\[
  \norm{Ch}_1
  \le
  2\int_{-\infty}^{\infty}
  \Var_{\mu_{R,r}}\!\parens*{\abs{S\cap A_s}}\,ds.
\]
\end{lemma}

\begin{proof}
Since \(C\) is a Laplacian,
\[
  (Ch)_i=\sum_{j\ne i}c_{ij}(h_i-h_j).
\]
Therefore
\[
  \norm{Ch}_1
  \le
  2\sum_{\substack{i<j\\ i,j\in R}}c_{ij}\abs{h_i-h_j}.
\]
The usual coarea identity for functions on a weighted graph gives
\[
  \sum_{\substack{i<j\\ i,j\in R}}c_{ij}\abs{h_i-h_j}
  =
  \int_{-\infty}^{\infty}
  \sum_{\substack{i\in A_s\\ j\notin A_s}}c_{ij}\,ds.
\]
The identity follows from the layer-cake formula
\(\abs{h_i-h_j}=\int \abs{\1\{h_i\ge s\}-\1\{h_j\ge s\}}\,ds\), summed with
weights \(c_{ij}\).
It remains to identify the cut with a variance.  Since \(C\1=0\),
\[
  \Var_{\mu_{R,r}}\!\parens*{\abs{S\cap A}}
  =
  \sum_{i,j\in A}C_{ij}
  =
  -\sum_{\substack{i\in A\\ j\notin A}}C_{ij}
  =
  \sum_{\substack{i\in A\\ j\notin A}}c_{ij}.
\]
Applying this with \(A=A_s\) proves the lemma.
\end{proof}

\begin{lemma}[One-step variance drop]\label{lem:variance-drop}
Fix an external-field measure \(\mu_{R,r}\) for which
\cref{conj:covariance-spectral} holds, let \(m=\abs{R}\), and fix
\(A\subseteq R\).  Let
\[
  Y=\abs{S\cap A},
  \qquad
  S\sim\mu_{R,r}.
\]
Reveal a uniformly random coordinate \(I\in R\) together with the value
\(\1_I(S)\), and let \(\cG\) be the revealed information.  Then
\[
  \E*{\Var(Y\mid\cG)}
  \le
  \parens*{1-\frac1m}\Var(Y).
\]
\end{lemma}

\begin{proof}
If \(r=0\) or \(r=m\), then \(Y\) is deterministic and the claim is trivial.
Assume \(0<r<m\).  Since the external field is positive, no coordinate is
deterministic.
Let \(C=C_{R,r}\), \(D=\diag(C_{ii})\), and \(a=\1_A\).  The law of total
variance gives
\[
  \Var(Y)-\E*{\Var(Y\mid\cG)}
  =
  \frac1m
  \sum_{i\in R}
  \frac{\Cov(\1_i,Y)^2}{\Var(\1_i)}.
\]
This is the expected amount of variance explained by observing one uniformly
random coordinate.  Since \(\Cov(\1_i,Y)=(Ca)_i\) and
\(\Var(\1_i)=D_{ii}\), the right-hand side is
\[
  \frac1m\, a^\top C D^{-1} C a,
\]
and \(D^{-1}\) is well-defined.  By \cref{conj:covariance-spectral}, the
nonzero eigenvalues of
\(D^{-1/2}CD^{-1/2}\) are at least \(1\).  Equivalently,
\[
  C D^{-1}C\succeq C.
\]
Indeed, after conjugating by \(D^{-1/2}\), this is the inequality
\(B^2\succeq B\), where the eigenvalues of \(B\) are \(0\) and numbers at
least \(1\).
Thus
\[
  \Var(Y)-\E*{\Var(Y\mid\cG)}
  \ge
  \frac1m\,a^\top C a
  =
  \frac1m\Var(Y),
\]
as claimed.
\end{proof}

\begin{lemma}[Dynamic-to-static bound]\label{lem:dynamic-static}
Assume \cref{conj:covariance-spectral} holds for every conditional
external-field measure reached by the sampler.  Then, for every \(h\),
\[
  \mathsf V(h)
  \le
  2\sum_{i<j}c_{ij}\abs{h_i-h_j},
\]
where \(c_{ij}=-C_{ij}\) for the initial covariance matrix \(C=C_{[n],k}\).
\end{lemma}

\begin{proof}
Apply \cref{lem:coarea} to each random state \((R_\ell,r_\ell)\):
\[
  \norm{C_{R_\ell,r_\ell}h_{R_\ell}}_1
  \le
  2\int
  \Var_{\mu_{R_\ell,r_\ell}}
  \!\parens*{\abs{S\cap A_s\cap R_\ell}}\,ds,
\]
where \(A_s=\set{i\in[n]\given h_i\ge s}\).  Therefore
\[
  \mathsf V(h)
  \le
  2\int
  \E*{
    \sum_\ell
    \frac1{\abs{R_\ell}}
    \Var_{\mu_{R_\ell,r_\ell}}
    \!\parens*{\abs{S\cap A_s\cap R_\ell}}
  }\,ds.
\]
In each variance above, \(S\) denotes a sample from the conditional
external-field measure appearing in the subscript.
For each fixed \(s\), \cref{lem:variance-drop} says that the variance of
\(\abs{S\cap A_s\cap R_\ell}\) drops, in expectation, by at least
\[
  \frac1{\abs{R_\ell}}
  \Var_{\mu_{R_\ell,r_\ell}}
  \!\parens*{\abs{S\cap A_s\cap R_\ell}}
\]
when the next uniformly random coordinate and its value are revealed.  Summing
over \(\ell\) telescopes:
\[
  \E*{
    \sum_\ell
    \frac1{\abs{R_\ell}}
    \Var_{\mu_{R_\ell,r_\ell}}
    \!\parens*{\abs{S\cap A_s\cap R_\ell}}
  }
  \le
  \Var_{\mu_{[n],k}}\!\parens*{\abs{S\cap A_s}}.
\]
Hence
\[
  \mathsf V(h)
  \le
  2\int
  \Var_{\mu_{[n],k}}\!\parens*{\abs{S\cap A_s}}\,ds.
\]
Using the coarea identity in the proof of \cref{lem:coarea} at the initial
state gives
\[
  \int
  \Var_{\mu_{[n],k}}\!\parens*{\abs{S\cap A_s}}\,ds
  =
  \sum_{i<j}c_{ij}\abs{h_i-h_j}.
\]
\end{proof}

\section{The Static Electrical Bound}\label{sec:static}

The covariance matrix \(C=C_{[n],k}\) is a Laplacian.  We keep the notation
\[
  c_{ij}=-C_{ij}\quad(i\ne j),
  \qquad
  d_i=C_{ii}=\sum_{j\ne i}c_{ij}.
\]
The following theorem is the static heart of the proof.  It is purely a
statement about weighted graphs.  The spectral assumption says that the
normalized adjacency matrix has no positive nontrivial eigenvalues.  For the
covariance Laplacians arising from external-field \(k\)-subset measures,
writing \(D=\diag(d_i)\) and \(W\) for the weighted adjacency matrix, this is
exactly the normalized covariance bound from \cref{conj:covariance-spectral},
since
\[
  D^{-1/2}WD^{-1/2}=I-D^{-1/2}CD^{-1/2}.
\]
\begin{theorem}[Unit-current bound]\label{thm:unit-current}
Assume \(C\) is any connected graph Laplacian with edge weights \(c_{ij}\),
degrees \(d_i\).  Let \(D=\diag(d_i)\), and let \(W\) be the weighted adjacency
matrix, with
\(W_{ij}=c_{ij}\) for \(i\ne j\) and \(W_{ii}=0\).  Assume the normalized
adjacency
\[
  D^{-1/2}WD^{-1/2}
\]
has Perron eigenvalue \(1\) and all other eigenvalues nonpositive.  Fix
distinct vertices \(a,b\), and let \(v\) be any solution of
\[
  Cv=\1_a-\1_b.
\]
Then
\[
  \sum_{i<j}c_{ij}\abs{v_i-v_j}\le 3.
\]
\end{theorem}

\begin{proof}
The solution \(v\) exists because \(C\) is the Laplacian of a connected graph:
its image is \(\set{z\given \sum_i z_i=0}\).  It is unique modulo multiples of
\(\1\), and all quantities below depend only on potential differences.

The statement is invariant under scaling \(C\) by a positive constant.  Indeed,
the normalized adjacency matrix does not change, the solution \(v\) scales by
the inverse constant, and \(\sum_{i<j}c_{ij}\abs{v_i-v_j}\) is unchanged.  We
may therefore assume
\[
  \sum_i d_i=1.
\]
Normalize the conductances into a reversible Markov chain.  Let
\[
  \pi_i=d_i,
  \qquad
  P_{ij}=\frac{c_{ij}}{d_i}\quad(i\ne j),
  \qquad
  P_{ii}=0.
\]
Then \(\pi\) is stationary for \(P\), and \(P\) is reversible because
\[
  \pi_iP_{ij}=c_{ij}=\pi_jP_{ji}.
\]
Thus the original edge weights are exactly the stationary (ergodic) edge flow
of the chain.  Equivalently,
\[
  C=\diag(\pi)(I-P).
\]
The matrix \(P\) is similar to \(D^{-1/2}WD^{-1/2}\), so its only positive
eigenvalue is the eigenvalue \(1\) on constants.  Write
\(\langle f,g\rangle_\pi=\sum_i\pi_i f_i g_i\).

Since \(Cv=\1_a-\1_b\), we have
\[
  (I-P)v=\eta,
  \qquad
  \eta=\frac{\1_a}{\pi_a}-\frac{\1_b}{\pi_b}.
\]
Define the flow-gradient seminorm
\[
  \norm{f}_{\nabla,1}
  =
  \sum_{i,j}\pi_iP_{ij}\abs{f_i-f_j}.
\]
The desired bound is equivalent to \(\norm{v}_{\nabla,1}\le 6\), because
\[
  \norm{v}_{\nabla,1}
  =
  \sum_{i,j}c_{ij}\abs{v_i-v_j}
  =
  2\sum_{i<j}c_{ij}\abs{v_i-v_j}.
\]
Replacing \(v\) by \(v+\alpha\1\) does not change either the equation or the
seminorm, so assume \(\sum_i\pi_i v_i=0\).  Put \(u=Pv\).  Then \(v=\eta+u\),
and the triangle inequality gives
\[
  \norm{v}_{\nabla,1}
  \le
  \norm{\eta}_{\nabla,1}+\norm{u}_{\nabla,1}.
\]
The first term is bounded by
\[
  \norm{\eta}_{\nabla,1}
  \le
  \sum_{i,j}\pi_iP_{ij}\parens*{\abs{\eta_i}+\abs{\eta_j}}
  =
  2\sum_i\pi_i\abs{\eta_i}
  =
  4.
\]

It remains to show \(\norm{u}_{\nabla,1}\le2\).  Consider the matrix
\[
  Q_{ij}=\pi_i\pi_j-\pi_iP_{ij}.
\]
We first note that \(Q\succeq0\).  Indeed, for any \(x\in\R^n\), let
\(\bar x=\sum_i \pi_i x_i\) and \(y=x-\bar x\1\).  Then
\[
  x^\top Qx
  =
  \bar x^2-\langle x,Px\rangle_\pi
  =
  -\langle y,Py\rangle_\pi
  \ge0,
\]
because \(P\) is nonpositive on the mean-zero subspace.  Thus there are vectors
\(z_i\) such that \(Q_{ij}=\langle z_i,z_j\rangle\).  Since
\(\sum_jQ_{ij}=0\), the vector \(\sum_jz_j\) is orthogonal to every \(z_i\).
It lies in their span, so \(\sum_i z_i=0\).  Define
\[
  r_i=\frac{z_i}{\pi_i}.
\]
Since \(P_{ii}=0\), the diagonal entries give
\(\norm{r_i}^2=Q_{ii}/\pi_i^2=1\).  Thus
\(\sum_i\pi_i r_i=0\), \(\norm{r_i}=1\), and
\[
  \pi_iP_{ij}=\pi_i\pi_j\parens*{1-\langle r_i,r_j\rangle}.
\]

Put
\[
  \ell=\sum_j\pi_j v_j r_j.
\]
Since \(v\) has \(\pi\)-mean zero,
\[
  u_i=(Pv)_i
  =
  \sum_j\pi_j\parens*{1-\langle r_i,r_j\rangle}v_j
  =
  -\langle r_i,\ell\rangle.
\]
Therefore
\[
  \norm{u}_{\nabla,1}
  =
  \sum_{i,j}\pi_i\pi_j\parens*{1-\langle r_i,r_j\rangle}
  \abs{\langle \ell,r_i-r_j\rangle}.
\]
By Cauchy--Schwarz, using
\(\sum_{i,j}\pi_i\pi_j(1-\langle r_i,r_j\rangle)=1\), which follows from
\(\sum_i\pi_i r_i=0\) and \(\norm{r_i}=1\), we get
\[
  \norm{u}_{\nabla,1}^2
  \le
  \sum_{i,j}\pi_i\pi_j\parens*{1-\langle r_i,r_j\rangle}
  \langle \ell,r_i-r_j\rangle^2.
\]
Let
\[
  \Sigma=\sum_i\pi_i r_i r_i^\top.
\]
For every vector \(x\),
\[
  x^\top\Sigma x
  =
  \sum_i\pi_i\langle x,r_i\rangle^2
  \le
  \sum_i\pi_i\norm{x}^2
  =
  \norm{x}^2,
\]
so \(0\preceq\Sigma\preceq I\).  Next write
\(a_i=\langle\ell,r_i\rangle\).  Then
\(\sum_i\pi_i a_i=0\), because \(\sum_i\pi_i r_i=0\), and
\(\ell^\top\Sigma\ell=\sum_i\pi_i a_i^2\).  Hence
\[
  \sum_{i,j}\pi_i\pi_j
  \langle \ell,r_i-r_j\rangle^2
  =
  \sum_{i,j}\pi_i\pi_j(a_i-a_j)^2
  =
  2\sum_i\pi_i a_i^2
  =
  2\ell^\top\Sigma\ell.
\]
Since \(r_i\) and \(r_j\) are unit vectors,
\(1-\langle r_i,r_j\rangle\le2\).  Therefore
\[
  \sum_{i,j}\pi_i\pi_j\parens*{1-\langle r_i,r_j\rangle}
  \langle \ell,r_i-r_j\rangle^2
  \le
  4\ell^\top\Sigma\ell.
\]
We claim that \(\ell^\top\Sigma\ell\le1\).  From \(v=\eta+u\),
\[
  \ell
  =
  \sum_i\pi_i \eta_i r_i+\sum_i\pi_i u_i r_i
  =
  r_a-r_b-\Sigma\ell.
\]
Thus \(\ell=(I+\Sigma)^{-1}(r_a-r_b)\), and hence
\[
  \ell^\top\Sigma\ell
  =
  (r_a-r_b)^\top(I+\Sigma)^{-1}\Sigma(I+\Sigma)^{-1}(r_a-r_b).
\]
Since \(\Sigma\succeq0\), its eigenvalues are nonnegative.  For every
\(\lambda\ge0\), \(\lambda/(1+\lambda)^2\le1/4\).  Therefore
\[
  (I+\Sigma)^{-1}\Sigma(I+\Sigma)^{-1}\preceq\frac14 I.
\]
Since \(\norm{r_i}=1\), \(\norm{r_a-r_b}\le2\), so
\[
  \ell^\top\Sigma\ell\le1.
\]
It follows that \(\norm{u}_{\nabla,1}\le2\), and therefore
\[
  \norm{v}_{\nabla,1}\le4+2=6.
\]
The equivalence above gives
\(\sum_{i<j}c_{ij}\abs{v_i-v_j}\le3\).
\end{proof}

\begin{corollary}[Arbitrary currents]\label{cor:arbitrary-current}
Let \(C=C_{[n],k}\) be the covariance matrix of a nondegenerate external-field
\(k\)-subset measure, with \(0<k<n\) and no deterministic coordinates, and
assume \cref{conj:covariance-spectral} holds for this measure.  For every
\(h\),
\[
  \sum_{i<j}c_{ij}\abs{h_i-h_j}
  \le
  \frac32\norm{Ch}_1.
\]
\end{corollary}

\begin{proof}
Let \(g=Ch\).  Since \(C\1=0\), \(\sum_i g_i=0\).  Decompose \(g\) into
elementary transfers
\[
  g=\sum_{a,b}\gamma_{ab}(\1_a-\1_b),
  \qquad
  \gamma_{ab}\ge0,
  \qquad
  \sum_{a,b}\gamma_{ab}=\frac12\norm{g}_1.
\]
For instance, send the positive part of \(g\) to the negative part.  For each
pair with \(\gamma_{ab}>0\), choose \(v_{ab}\) satisfying
\(Cv_{ab}=\1_a-\1_b\).  The two sides of the next display both solve \(Cu=g\),
so they differ by a multiple of \(\1\):
\[
  h=\sum_{a,b}\gamma_{ab}v_{ab}
  \qquad\text{modulo multiples of \(\1\)}.
\]
The map \(\Phi(u)=\sum_{i<j}c_{ij}\abs{u_i-u_j}\) is a seminorm, and
\cref{thm:unit-current} applies by \cref{conj:covariance-spectral}.  Thus
\[
  \Phi(h)
  \le
  \sum_{a,b}\gamma_{ab}
  \Phi(v_{ab})
  \le
  3\sum_{a,b}\gamma_{ab}
  =
  \frac32\norm{Ch}_1.
\]
\end{proof}

\section{Proof of the Conditional Main Theorems}\label{sec:proof-main}

\begin{proof}[Proof of \cref{thm:main}]
The case \(k=n\) is deterministic, so assume \(0<k<n\).  Work under the
hypothesis in the statement of \cref{thm:main}.
Let \(x\in\relint(\Delta_{n,k}^{=})\), and let \(h\) be an infinitesimal
external-field perturbation at \(x\).  Combining \cref{lem:dynamic-static} and
\cref{cor:arbitrary-current} gives
\[
  \mathsf V(h)
  \le
  2\sum_{i<j}c_{ij}\abs{h_i-h_j}
  \le
  3\norm{Ch}_1.
\]
By \cref{lem:infinitesimal-reduction}, the infinitesimal stretch is at most
\(6\).  Integrating along the line segment from \(x\) to \(y\) proves
\[
  \E*{\abs{A(x)\triangle A(y)}}\le 6\norm{x-y}_1
\]
for relative-interior points.  Boundary points follow by deleting forced-zero
coordinates, contracting forced-one coordinates, and taking limits.
\end{proof}

\begin{proof}[Proof of \cref{cor:at-most}]
For \(x\in\Delta_{n,k}^{\le}\), let
\[
  \sigma(x)=k-\sum_{i=1}^n x_i.
\]
Encode the slack using \(k\) dummy coordinates by the monotone prefix vector
\(d(\sigma)\in[0,1]^k\): the first \(\lfloor\sigma\rfloor\) entries are \(1\),
the next entry is the fractional part of \(\sigma\), and the remaining entries
are \(0\).  Then
\[
  \sum_{\ell=1}^k d_\ell(\sigma)=\sigma,
  \qquad
  \norm{d(\sigma)-d(\tau)}_1=\abs{\sigma-\tau}.
\]
The extended vector
\[
  \widehat x=(x,d(\sigma(x)))
  \]
lies in \(\Delta_{n+k,k}^{=}\).  Apply \cref{thm:main} to the exact
hypersimplex on the real and dummy coordinates, and delete the dummy
coordinates from the rounded set.  Real marginals are preserved, and the real
symmetric difference is at most the full symmetric difference.  Therefore
\[
  \E*{\abs{A(x)\triangle A(y)}}
  \le
  6\norm{\widehat x-\widehat y}_1
  =
  6\parens*{\norm{x-y}_1+\abs{\sigma(x)-\sigma(y)}}.
\]
Since
\[
  \abs{\sigma(x)-\sigma(y)}
  =
  \abs{\sum_i(y_i-x_i)}
  \le
  \norm{x-y}_1,
\]
the stretch is at most \(12\).
\end{proof}

\section*{Acknowledgments}

We thank Aniket Das for helpful discussions.  We thank Aram Harrow for
pointing out the error in an earlier version's proof of the covariance
spectral assertion now stated as \cref{conj:covariance-spectral}.

\appendix

\section{Provenance}\label{app:provenance}

This manuscript used GPT 5.5 Pro Extended and Codex in a substantive way.  The
authors proposed the maximum-entropy distribution with prescribed hypersimplex
marginals, the sequential random-order sampler, and the infinitesimal
perturbation viewpoint.

A previous version claimed an unconditional proof of the constant bound.  That
version included an incorrect proof of the covariance spectral assertion now
stated as \cref{conj:covariance-spectral}.  The present manuscript keeps the
variance-drop reduction, the static \(L^1\) current estimate, and the
conditional implication from the spectral hypothesis to constant stretch.
Codex was used to help assemble the manuscript, add surrounding exposition and
references, and typeset the draft under the authors' supervision.

\PrintBibliography

\end{document}